\documentclass{cccg23}
\usepackage{tda}

\title{Metric and Path-Connectedness Properties of the Fr\'echet Distance for Paths and Graphs}

\author{Erin Chambers\thanks{Department of Computer Science,
        Saint Louis University, {\tt erin.chambers@slu.edu}}
        \and
        Brittany Terese Fasy\thanks{School of Computing,  Department of Mathematics, Montana State University {\tt  brittany.fasy@montana.edu}}
        \and
        Benjamin Holmgren \thanks{School of Computing, Montana State University {\tt  benjamin.holmgren1@student.montana.edu}}
        \and
        Sushovan Majhi \thanks{School of Information,  University of California, Berkeley {\tt  smajhi@ischool.berkeley.edu}}
        \and
        Carola Wenk \thanks{Department of Computer Science,  Tulane University {\tt  cwenk@tulane.edu}}
        }

\begin{document}

\maketitle

\begin{abstract}
    The Fr\'echet distance is often used to measure distances between paths, with
applications in areas ranging from map matching to GPS trajectory analysis to 
hand-writing recognition.  More recently, the Fr\'echet
distance has been generalized to a distance between two copies of the same
graph embedded or immersed in a metric space; this more general setting opens 
up a wide range of more complex applications in graph analysis.  In this paper, we 
initiate a study of some of the fundamental 
topological properties of spaces of paths and of graphs mapped to $\R^n$ under
the Fr\'echet distance, in an effort to lay the theoretical groundwork for
understanding how
these distances can be used in practice.  In particular, we prove whether or not these spaces, and
the metric balls therein, are path-connected.

\end{abstract}

\journal{path-connected space of paths. path-vs-map-vs-curve consistency - BTF}
\journal{check for: ``it is'', ``It is''}
\journal{BTF: we need to check for alarming phrasing, including `easy'}

\section{Introduction}
\label{sec:intro}
One-dimensional
data in a Euclidean ambient space is heavily studied in the computational geometry
literature, and is central to
applications in GPS trajectory and road network analysis~\cite{Ahmed2012, mapmatching, Guibas2011, qarta}.
One widely used distance measure on one-dimensional data is the  Fr\'echet distance,
which accounts for both geometric closeness as well as the connectivity of the paths or graphs being
compared~\cite{Frechet, Fang2021, buchin2010frechet, bos-cfdrvs-17, Buchin2020, Alt1995, Aronov2006,Agarwal2014,Driemel2012,Alt2001,Driemel2013,Cook2010,CHAMBERS2010295,Driemel2022,mapmatching,Colombe2021,Guibas2011}.
We build a theoretical foundation for these application areas by
investigating spaces of paths and graphs in~$\R^n$, including their metric and
topological properties, under the Fr\'echet distance.
The motivation for this work is simple: as practical approaches to compute the Fr\'echet distance between
paths~\cite{Aronov2006, Colombe2021} and between graphs~\cite{Driemel2022, Fang2021, Buchin2020} grow in popularity,
it is natural to inquire about the fundamental properties of such distances, in an effort to better understand
exactly what they are capturing.

We begin by defining the Fr\'echet distance between paths and
graphs.
Using open balls under the Fr\'echet distance to generate
a topology, we study the metric and topological properties of the
induced spaces.
In particular, we work with three classes of
paths: the space~$\cpaths$ of all paths in~$\R^n$, the space
$\epaths$ of all paths in~$\R^n$ that are embeddings (i.e., maps that are
homeomorphisms onto the image), and the space $\ipaths$ of all paths in~$\R^n$
that are immersions (local embeddings).
See \figref{path-sets} for examples of paths in~$\R^2$.
In addition, we study the three analogous spaces of graphs:
the sets~$\cgraphs$,~$\igraphs$, and~$\egraphs$ of continuous maps, immersions, and embeddings of graphs, respectively.
This paper establishes the core metric and topological properties of the
Fr\'echet distance on graphs and
paths in Euclidean~space.

\begin{figure}[bht]
    \centering
    \begin{subfigure}[b]{0.12\textwidth}
        \includegraphics[width=\textwidth]{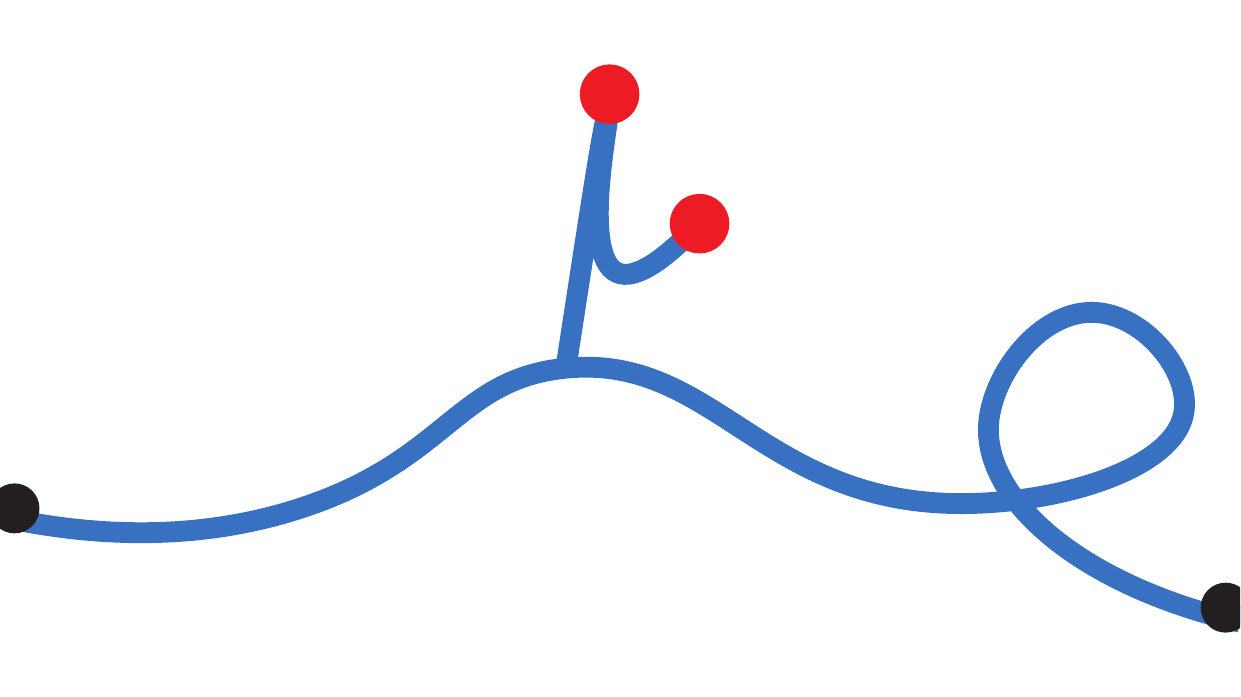}
        \caption{Continuous}
        \label{fig:imm-path-interp}
    \end{subfigure}
    ~ 
    \begin{subfigure}[b]{0.12\textwidth}
        \includegraphics[width=\textwidth]{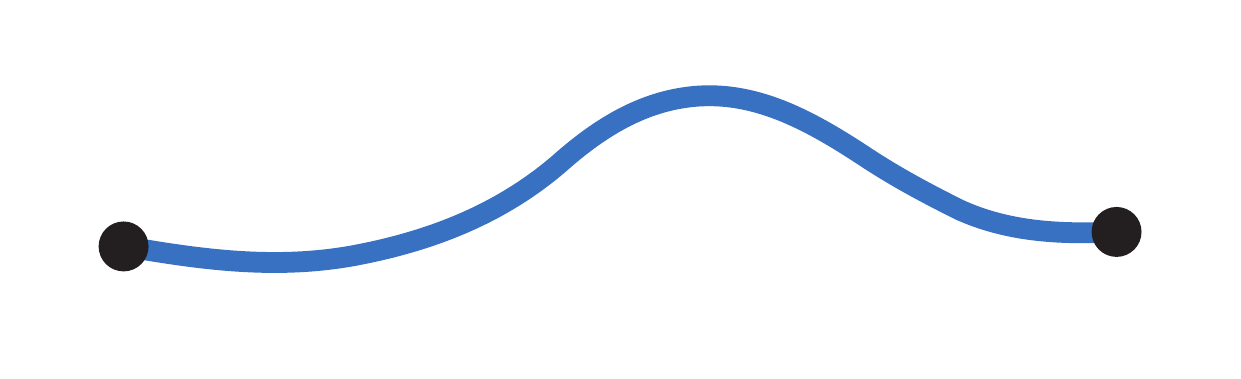}
        \caption{Embedding}
        \label{fig:imm-delta}
    \end{subfigure}
    ~
    \begin{subfigure}[b]{0.12\textwidth}
        \includegraphics[width=\textwidth]{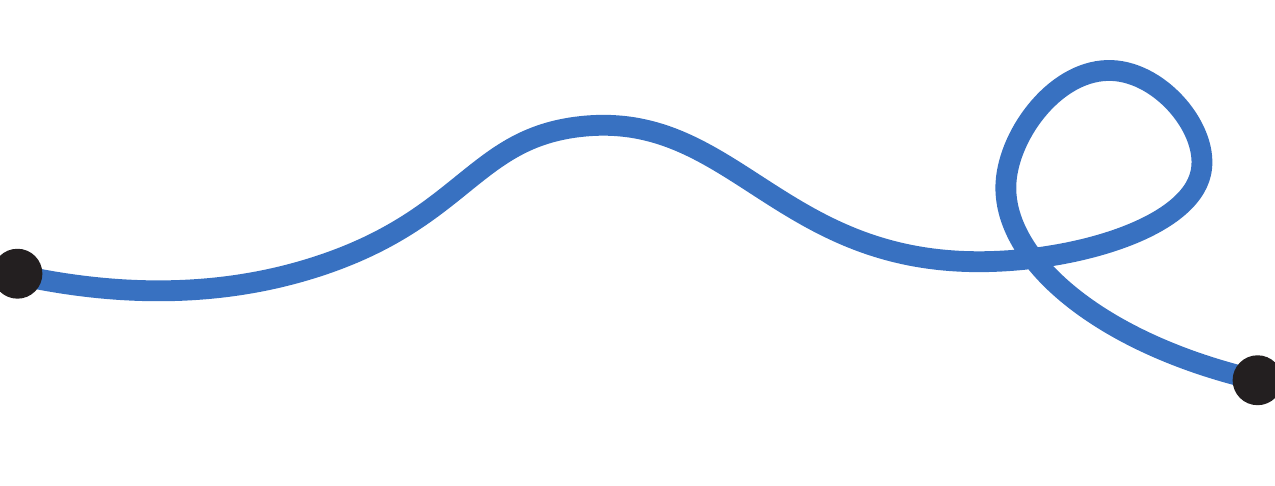}
        \caption{Immersion}
        \label{fig:imm-inflate}
    \end{subfigure}
    \caption{Example of paths continuously mapped, embedded, and immersed in
    $\R^2$.
    The space of continuous maps allows arbitrary self-intersection on a path including
    backtracking (which occurs at the two red points);
    embeddings must induce homeomorphisms onto their
    image; and immersions are locally embeddings.
    }\label{fig:path-sets}
\end{figure}

\section{Background}\label{sec:background}
In this section, we establish the definitions and notation from geometry and
topology used throughout. We assume basic knowledge of concepts in topology.
For common definitions central to this paper, we refer readers
to \appendref{basics}, or for greater detail, to~\cite{Edelsbrunner,rogue-1}.

\begin{definition}[Types of Maps]
    Let $\X$ and $\Y$ be topological spaces.
    A map $\alpha \colon \X \to \Y$ is called \emph{continuous} if for each open
    set~$U \subset \Y$, $\alpha^{-1}(U)$ is open in $\X$. We call~$\alpha$ an
    \emph{embedding} if
    $\alpha$ is injective. Equivalently, an \emph{embedding} is a continuous
    map that is
    homeomorphic onto its image.  If $\alpha$ is locally an embedding, then we say
    that~$\alpha$ is an~\emph{immersion}.
\end{definition}

In particular, a continuous map $\gamma \colon [0,1] \to \R^n$ is called a
\emph{path} in $\R^n$.
We call a path $\gamma \colon [0,1] \to \R^n$
\emph{rectifiable} if $\gamma$ has finite \emph{length} (see \defref{length} in
\appendref{distdefs}).  Moreover, we call a graph~$G$
rectifiable if there exists a finite cover of $G$ such that
every element in the cover is a rectifiable~path.

\paragraph{Paths in $\R^n$}
Letting $\elems{\cpaths}$ denote
the set of all rectifiable paths in $\R^n$,
we now define the path Fr\'echet distance.
\begin{definition}[The Path Fr\'echet Distance~\cite{Alt1995}]
    The \emph{Fr\'echet distance}
    $\comparefcn{d_{FP}}{\elems{\cpaths}}$
    between~$\apath_1,\apath_2 \in \elems{\cpaths}$ is defined~as:

    \[   d_{FP}(\apath_1,\apath_2) := \inf_{r \colon [0,1] \to [0,1]} \max_{t \in [0,1]} ||
        \apath_1(t) -
        \apath_2(r(t)) ||_2,
    \]
    where $r$ ranges over all
    homeomorphisms such that~$r(0)=0$, and $||\cdot||_2$ denotes the Euclidean norm.
\end{definition}

\paragraph{Graphs Mapped to $\R^n$}
We define a
\emph{graph}~$G=(V,E)$ as a finite set of vertices $V$ and a finite set of edges~$E$.
Self-loops and multiple edges between a pair of vertices are
allowed.\footnote{Some references would call this a \emph{multi-graph}, but for
simplicity, we just use the term \emph{graph}.}
We topologize a graph by thinking of it as a CW complex;
see \appendref{graphs}.
If $\phi \colon G \to \R^d$ is a map, then we call
$(G,\phi)$ a graph-map pair.
We extend the path Fr\'echet distance to the Fr\'echet distance between
graphs continuously mapped into $\R^n$:

\begin{definition}[Graph Fr\'echet Distance]\label{def:graphfrech}
    Let~$(G,\phi),(H,\psi)$ be continuous, rectifiable graph-map pairs.
    We define the Fr\'echet distance between $(G,\phi)$
    and~$(H,\psi)$ by minimizing over all homeomorphisms:\footnote{Other generalizations
    of the Fr\'echet distance minimize over all ``orientation-preserving''
    homeomorphisms, which can be defined in several ways for stratified spaces,
    and sometimes adding an orientation is not natural.
    Thus, we drop this requirement in our definition.}
    \[
        d_{FG}\left( (G,\phi),(H,\psi)\right) :=
            \begin{cases}
                \inf_{h} || \phi - \psi\circ h ||_{\infty} & G \cong H.\\
                \infty & \text{otherwise}.\\
            \end{cases}
    \]
    For simplicity of exposition, when $G \cong H$,
    we write the LHS of this equation as~$d_{FG}(\phi,\psi)$. Furthermore,
    defining the infimum over an emptyset to be $\infty$,
    the graph Fr\'echet distance is given by the following equation:
    \[
        d_{FG}\left( (G,\phi),(H,\psi)\right) :=
            \inf_{h} || \phi - \psi \circ h ||_{\infty},
    \]
    where the infimum is taken over all homeomorphisms~$h \colon G \to H$.
\end{definition}

Note that if $G=H$ and $\phi$ is a reparameterization of~$\psi$,
then~$d_{FG}(\phi,\psi)=0$.

\begin{obs}[Paths as Graphs]\label{obs:path-graph}
    If $G=[0,1]$ and~$\alpha,\beta \colon [0,1] \to \R^n$ are paths,
    then the relationship between path and graph Fr\'echet distances
    is as follows:
    \[
        d_{FG} \left( \alpha, \beta \right) =
        \min \left\{ d_{FP}(\alpha,\beta), d_{FP}(\alpha,\beta^{-1}) \right\},
    \]
    where $\beta^{-1}\colon I \to \R^n$ is defined by $\beta^{-1}(t) = \beta(1-t)$.
\end{obs}

\section{Metric Properties}
We now address the question:
Is this distance
a metric? If not, can it be metrized?
A well-known known property of the path Fr\'echet distance is that it
is a pseudo-metric~\cite{Frechet,Alt1995}.
That is, it satisfies all metric properties except for separability.
We proof  this property for $d_{FG}$.

\begin{theorem}[Metric Properties of $d_{FG}$]\label{thm:cont-metric-prop}
    $d_{FG}$ is an extended pseudo-metric that does not satisfy
    separability.  When restricted to a homeomorphism class of graphs, $d_{FG}$
    is a pseudo-metric.
\end{theorem}

\begin{proof}
    We first prove that $d_{FG}$ is an extended
    pseudo-metric (see \defref{pseudometric} in \appendref{distdefs}).

    Identity: Taking $h$ to be
    the identity map in \defref{graphfrech}, we find~$d_{FG}((G,\agraph_1),
    (G,\agraph_1))=0$.

    Symmetry:
    Consider~$d_{FG}(\agraph_1, \agraph_2)$.
    If $G \not\cong H$, then no homeomorphism $h \colon G \to H$ exists.  Likewise, no homeomorphism $h'
    \colon H \to G$ exists.  And, so,
    $$d_{FG}((G,\agraph_1), (H,\agraph_2)))=\infty = d_{FG}((H,\agraph_2),
    (G,\agraph_1))).$$
    Otherwise, since $h$ is a homeomorphism, it is invertible.  Thus, we can rewrite this
    as:
    $$ d_{FG}(\agraph_1, \agraph_2) =
    \underset{h^{-1}}{\inf} || \phi_1 \circ h^{-1} -
    \phi_2||_{\infty}=d_{FG}(\agraph_2, \agraph_1).$$

    Subadditivity (the triangle inequality):
    Consider~$d_{FG}((G_1,\phi_1),(G_2,\phi_2)) +
    d_{FG}((G_2,\phi_2),(G_3,\phi_3))$.  If $G_1 \not\cong G_2$, then
    $d_{FG}((G_1,\phi_1),(G_2,\phi_2))=\infty$, and we are done.
    A symmetric argument follows for $G_2 \not\cong
    G_3$. Thus, we assume $G_1 \cong G_2 \cong G_3$.
    Using the definition of Fr\'echet distance and the fact that the infimum is
    taken over homeomorphisms, we obtain:
    \begin{align*}
        &d_{FG}(\phi_1,\phi_2) + d_{FG}(\phi_2,\phi_3)\\
        &\qquad = \inf_{h'} || \phi_1 - \phi_2 \circ h'  ||_{\infty}
            + \inf_{h''} ||  \phi_2 - \phi_3\circ h'' ||_{\infty}.\\
        &\qquad \geq \inf_{h'} || \phi_1 - \phi_2 \circ h'  ||_{\infty}
            + \inf_{h,h'} ||  \phi_2 \circ h' - \phi_3\circ h ||_{\infty}\\
        &\qquad = \inf_{h,h'} || \phi_1 + (\phi_2 \circ h'-\phi_2 \circ h') - \phi_3\circ h ||_{\infty}\\
        &\qquad = \inf_{h} || \phi_1 - \phi_3\circ h ||_{\infty}\\
        &\qquad = d_{FG}(\phi_1,\phi_3).
    \end{align*}
    And so, we conclude
    that~$d_{FG}$ satisfies subadditivity.

    Noting that if $G \not\cong H$ that $d_{FG}\left( (G,\phi),(H,\psi)\right) =
    \infty$, we conclude that~$d_{FG}$ is an extended pseudo-metric.
    However, the graph Fr\'echet distance between homeomorphic graphs
    is at most the Hausdorff distance between the
    images of the two maps.
    Thus, when restricted to a homeomorphism class of graphs, $d_{FG}$
    is a pseudo-metric.
\end{proof}

The only metric property not satisfied is separability.
\journal{proof of not separable:
    Finally, we prove $d_{FG}$ does not satisfy separability.
    Let~$G$ be a graph such that $G \equiv \S^1$. Let~$\agraph:G \to \R^2$ be a
    homeomorphism such that $\agraph(G)=\S^1$.
    Let the function~$f \colon \S^1 \to \S^1$
    be defined by~$f\left( e^{i\theta}\right) =  e^{i(\theta+\pi)} $.
    Then, we know
    that~$\phi \circ f \neq \phi$, but~$d_{FG}(\phi,\phi\circ f)=0$.
}
In order to metrize this pseudo-metric, we define~$\cgraphs(G)$ to be the
the set of equivalence classes of continuous, rectifiable maps $G \to \R^n$,
where two maps,~$\agraph_1$ and~$\agraph_2$, are equivalent
if and only
if~$d_{FG}(\agraph_1, \agraph_2) = 0$. We write~$[\agraph_i]$ to denote the equivalence
class of maps containing $\agraph_i$. We define two subspaces of $\cgraphs(G)$: those representing
immersions and embeddings, denoted~$\igraphs(G)$ and $\egraphs(G)$,
respectively.
Note that~$\egraphs(G) \subsetneq \igraphs(G) \subsetneq \cgraphs(G)$.  Let $\cgraphs$
denote the induced set of equivalence classes of all graph-map
pairs~$(G,[\agraph])$ such that~$\eqgraph{} \in
\cgraphs(G)$.  Similiarly, we define $\igraphs$ and $\egraphs$, and
note~$\egraphs \subsetneq \igraphs \subsetneq \cgraphs$. Hence,

\begin{cor}[Metric Extension for Graphs]\label{cor:metric-graph}
    For every graph~$G$, the
    graph Fr\'echet distance is a metric on the quotient spaces
    $\cgraphs(G)$, $\igraphs(G)$, and~$\egraphs(G)$.
    Moreover, the graph Fr\'echet distance is an extended metric
    on~$\cgraphs$,~$\igraphs$, and~$\egraphs$.
\end{cor}

Similarly, we consider paths in~$\R^n$: in particular, $\cpaths$ is the
set of equivalences classes of $\elems{\cpaths}$ up to orientation-preserving
reparameterization.  Equivalently, for
$\apath_1, \apath_2 \in \elems{\cpaths}$, $\apath_1$ is equivalent to $\apath_2$
iff~$d_{FP}(\apath_1, \apath_2) = 0$.
Likewise,~$\epaths$ and $\ipaths$ are the subspaces  of embeddings and
immersions.
Note that~$\epaths \subsetneq \ipaths \subsetneq \cpaths$.
We topologize these spaces using the open ball topology (\appendref{distdefs}).
Again, by construction, we~obtain:

\begin{cor}[Metric Properties of $d_{FP}$]
    The path Fr\'echet distance is a
    metric on $\cpaths, \ipaths$ and~$\epaths$.
\end{cor}

\section{Path-Connectedness Property}\label{sec:path-conn-sec}
We now examine path-connectedness properties. See \defref{path-con}
and \defref{ball-con} of \appendref{path} for definitions of path-connectivity.

\subsection{Continuous Mappings}\label{sec:cont}
We start with the most general spaces of paths and graphs:
the continuous, rectifiable  maps into $\R^n$.
In Euclidean spaces, linear interpolation is a useful tool because it defines
the shortest paths between two points.  In function spaces, linear interpolation
is also nice:

\begin{definition}[Linear Interpolation]\label{def:interp-graphs}
    Let $G$ be a graph and $\agraph_0, \agraph_1 \colon G \to \R^n$ be
    continuous, rectifiable maps.
    The \emph{linear interpolation} from~$\agraph_{0}$ to $\agraph_{1}$ is
    the map~$\Gamma \colon [0,1] \to \cgraphs(G)$ sending $t\in[0,1]$
    to~$(G,\agraph_t)$,~where:
    \begin{equation}\label{eq:lin-interp}
        \agraph_t := (1-t)\agraph_0 + t(\agraph_1 \circ h_*).
    \end{equation}
    For ease of notation, we sometimes write $\Gamma_t := \Gamma(t)$.
\end{definition}

Note that $(1-t)\agraph_0 + t\agraph_1$ is a linear combination of~$\agraph_0$ and
$\agraph_1$ (using $c_0=1-t$ and $c_1=t$ in \defref{combo}).
Thus,~$\Gamma$ is a continuous family of linear combinations
of the maps
$\agraph_0$ and~$\agraph_1$; we show $\Gamma$ is continuous in \lemref{gamma-cont}.
in \appendref{lin-interp}. If $G=[0,1]$, the linear interpolation between
graphs is simply linear interpolation between paths.
For an example of linear interpolation between graphs, see
\figref{interp} in \appendref{lin-interp}.

However, linear interpolation is not well-defined in~$\cgraphs$, as we could
have~$\agraph_1,\agraph_2 \in \eqgraph{} \in \cgraphs(G)$.
In fact,~$\Gamma(t; \agraph_1,\agraph_2) = \Gamma(t; \agraph_1,\agraph_3)$ if and only if
$\agraph_1=\agraph_2$.

\begin{definition}[Family of Interpolations]
    Let $G$ be a graph and $\eqgraph{0}, \eqgraph{1} \in \cgraphs(G)$.
    We define~$\mapclass{\eqgraph{0}}{\eqgraph{1}}$ to be the set of
    all linear interpolations
    between elements of $\eqgraph{0}$ and of $\eqgraph{1}$.
\end{definition}

We now demonstrate the existence of a family of interpolations
between any two equivalence classes within~$(\cgraphs(G), d_{FG})$,
proving path-connectivity.

\begin{theorem}[Continuous Maps of Graphs]\label{thm:graph-cont}
    For every graph~$G$, the extended metric space $(\cgraphs(G), d_{FG})$ is path-connected.
    Moreover, the connected components of $(\cgraphs,d_{FG})$
    are in one-to-one correspondence with the homeomorphism classes of~graphs.
\end{theorem}

\begin{proof}
    Let $\eqgraph{0}, \eqgraph{1} \in \cgraphs(G)$.
    Let~$\Gamma \in \mapclass{\eqgraph{0}}{\eqgraph{1}}$.
    By
    \lemref{gamma-cont} in \appendref{lin-interp},~$\Gamma$ is continuous,
    and so~$(\cgraphs(G), d_{FG})$ is path-connected.

    Moreover, suppose~$(G,\eqgraph{0}), (H,\eqgraph{1}) \in \cgraphs$ for the graphs $G, H$ which are not homeomorphic.
    Then,~$d_{FG}((G,\eqgraph{0}), (H,\eqgraph{1})) = \infty$, and connected components
    of the extended metric space $\cgraphs$ are vacuously in one-to-one correspondence with
    homeomorphism classes of graphs.
\end{proof}

Setting $G = [0,1]$, an identical proof holds for paths.

\begin{cor}[Continuous Maps of Paths]\label{cor:path-cont}
    The space $\cpaths$ is path-connected.
\end{cor}

We now demonstrate the stricter property of the path-connectivity of open
distance balls:

\begin{lemma}[Metric Balls in $(\cgraphs, d_{FP})$]\label{lem:cont-graph-balls}
    Metric balls with finite radius in~$(\cgraphs, d_{FP})$ are path-connected.
\end{lemma}

\begin{proof}
    Let $\delta \in \R$ such that $\delta > 0$.
    Let $(G,\eqgraph{0}) \in \cgraphs$.

    Consider the metric ball $\B := \ball{d_{FG}}{\eqgraph{0}}{\delta}$
    in $\cgraphs$.
    Let~$\eqgraph{1},\eqgraph{2} \in \B$. We wish to find a path from
    $\eqgraph{1}$ to $\eqgraph{2}$.
    We first find a path in~$\ball{d_{FG}}{\eqgraph{0}}{\delta}$
    from~$\eqgraph{0}$ to~$\eqgraph{2}$, as follows.
    Set $$\eps = \delta - d_{FG}(\eqgraph{0}, \eqgraph{2}).$$
    By \lemref{hexist}, we know that there exists a homeomorphism $h_*:G \to G$
    such that the following inequality holds:~$||\agraph_0 - \agraph_2 \circ
    h_*||_{\infty}<d_{FG}(\eqgraph{0}, \eqgraph{2}) + \eps/2$.

    Let $\Gamma \in \mapclass{\eqgraph{0}}{\eqgraph{2}}$.
    Then, for all $t \in (0,1)$,
    \begin{align*}
        & d_{FG}(\Gamma_t, \agraph_0) \\
        & \qquad =\inf_h || ((1-t)\agraph_0 + t(\agraph_2\circ h_*)) - \agraph_0 \circ h||_{\infty} \\
        & \qquad \leq || ((1-t)\agraph_0 + t(\agraph_2\circ h_*)) - \agraph_0 \circ h_*||_{\infty} \\
        \journal{
        & \qquad \leq ||\agraph_0 - \agraph_2 \circ h_*||_{\infty}
            + || -t\agraph_0 + t(\agraph_2\circ h_*)) ||_{\infty}\\
        }
        & \qquad < d_{FG}(\eqgraph{0}, \eqgraph{2}) + \eps/2\\
        & \qquad < \delta.
    \end{align*}
    \journal{BTF: where does $h_*$ come from?}
    Thus, $\Gamma_t \in \ball{d_{FG}}{\eqgraph{1}}{\delta}$, which
    means there exists a path from~$\agraph_0$ to~$\agraph_2$.  Similarly, we
    find a path $\Gamma'$ from~$\agraph_1$ to~$\agraph_0$.  Concatentating the two
    paths,~$\Gamma' \# \Gamma$
    we have a path in~$\ball{d_{FG}}{\eqgraph{0}}{\delta}$
    from~$\eqgraph{1}$ to~$\eqgraph{2}$.
    Hence, metric balls with finite radius in $\cgraphs$ are path-connected.
\end{proof}

Setting $G=[0,1]$, we obtain:

\begin{cor}[Metric Balls in $(\cpaths, d_{FP})$]\label{cor:cont-path-balls}
    Balls in the extended metric space $(\cpaths, d_{FP})$ are path-connected.
\end{cor}

\subsection{Immersions}\label{sec:imm}
An immersion is a map that is locally injectivite. Thus, self-intersections are allowed, but a map
pausing or backtracking is not. Next, we define these notions,
and give examples in \figref{bad-immerse}.

\begin{definition}[Pausing]\label{def:pause}
    We say that a path $\apath$ \emph{pauses} in an interval $I \subset
    [0,1]$ if~$\apath(x) = \apath(y)$ for every~$x, y \in I$.
    In this case, $[\apath] \not \in \ipaths$.
\end{definition}

Another possible violation of local injectivity is \emph{backtracking} on a path.

\begin{definition}[Backtracking]\label{def:backtrack}
    We say that a path $\apath$ is \emph{backtracking} at a point $x \in [0,1]$ if there exists $\delta > 0$ such that
    for every $\epsilon \in (0,\delta)$, either
    $\apath|_{(x - \epsilon, x)} \subset \apath_{(x, x + \epsilon)}$ or~$\apath|_{(x, x + \epsilon)} \subset \apath|_{(x - \epsilon, x)}$.
\end{definition}

To show the path-connectivity of spaces of immersions, the proof
in \thmref{graph-cont} for continuous mappings is
\emph{almost} sufficient, but these added violations must be addressed.
Thus, we introduce additional maneuvers to
avoid pauses and backtracking.
\begin{figure}
    \centering
    \begin{subfigure}[b]{0.2\textwidth}
        \includegraphics[width=\textwidth]{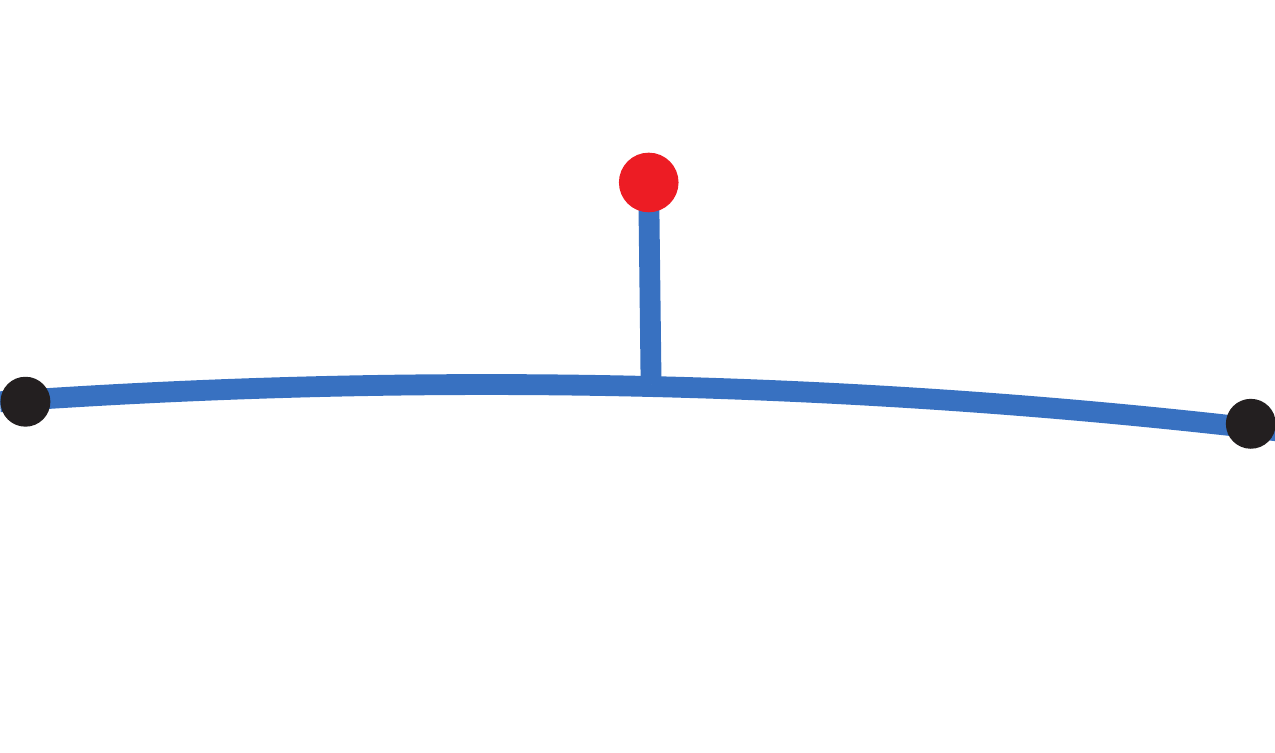}
        \caption{Forced Backtracking}
        \label{fig:backtrack}
    \end{subfigure}
    ~ 
    \begin{subfigure}[b]{0.2\textwidth}
        \includegraphics[width=\textwidth]{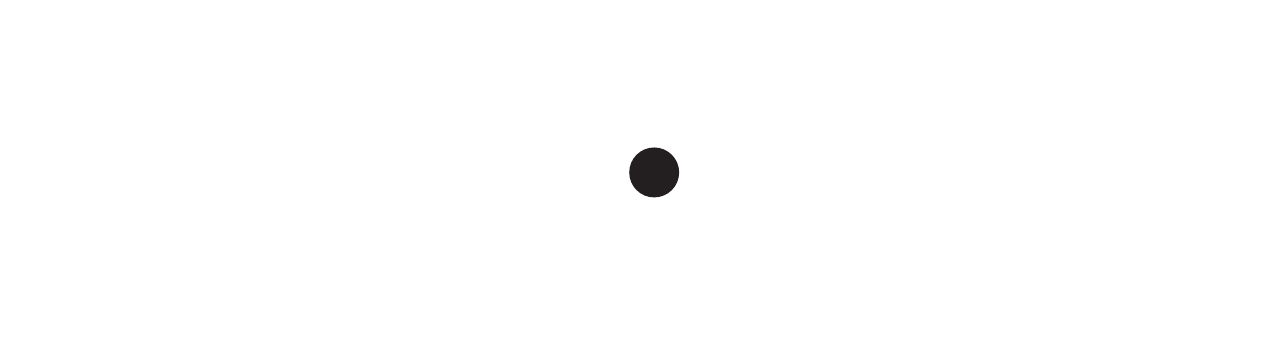}
        \caption{Constant Map}
        \label{fig:point}
    \end{subfigure}
    \caption{
	Examples of paths in $\cpaths$ but not $\ipaths$. \figref{backtrack} demonstrates
	a path with necessary backtracking
	at the red point. \figref{point} demonstrates a constant path which (vacuously) must pause.
    For a nontrivial example of a path with pauses, consider
	any parameterization of a path sending an open interval to a point.
    }\label{fig:bad-immerse}
\end{figure}

\begin{lemma}[Rerouting Pauses]\label{lem:thepause}
    Let~$\apath_{0}, \apath_{1} \in \elems\ipaths$, and let
    $\Gamma:[0,1] \to \elems\cpaths$ be a path in $\elems\cpaths$
    from~$\apath_{0} $ to~$\apath_{1}$.
    Suppose there exists an interval~$[t_1,t_2]$ such that for all~$t \in
    [0,1]\setminus (t_1,t_2)$, $\Gamma_t$ is an immersion.  But, for all $t
    \in~(0,1)$, $\Gamma_t$ has a single pause.
    Then, there exists a different path~$\Gamma^* \colon [0,1] \to
    \ipaths$ that avoids the~pause.
\end{lemma}
\begin{proof}
    Let $t \in [t_1,t_2]$. Let the pause in $\Gamma_t$ be over the
    interval~$(a_t,b_t)\subset[0,1]$.
    Let $\eps_t:=\min(t_2-t,t-t_1)$.
    \journal{BTF: we also need to account for `going off the edge'}
    We stretch the paused interval $(a_t,b_t)$ in~$\Gamma_t$ by defining a
    map~$\Gamma^*_t:[0,1]\to\elems\ipaths$~as follows:
    \begin{itemize}
        \item $\Gamma_t^*(-\infty,a_t]$ is an oriented reparameterization
            of~$\Gamma_t(-\infty,a_t-\eps_t]$.
        \item $\Gamma_t^*(a_t,b_t)$ is an oriented reparameterization
            of~$\Gamma_t(a_t-\eps_t,a_t] \# \Gamma_t[b_t,b_t+\eps)$
        \item $\Gamma_t^*[b_t,\infty)$ is an oriented reparameterization
            of~$\Gamma_t[b_t+\eps_t,\infty)]$.
    \end{itemize}
    \journal{more precisly ...
    \begin{equation}\label{eq:reparam-1}
        \Gamma^*_t(x) :=
        \begin{cases}
            \Gamma_t\left(\frac{a_tx}{a_t-\eps_t} \right) & \text{if } x \in
            [0,a_t-\eps_t]\\
            \todo{\Gamma_t(2(x-b)\cdot(1-b) + b)} & \text{if } x \in (a_t-\eps_t,\frac{a_t+b_2}{2})\\
            \todo{\Gamma_t(2(x-b)\cdot(1-b) + b)} & \text{if } x \in (\frac{a_t+b_2}{2},b_t+\eps_t)\\
            \todo{\Gamma_t(2(x-b)\cdot(1-b) + b)} & \text{if } x \in [b_t+\eps_t,1]
        \end{cases}
    \end{equation}
    }
    By construction, $\Gamma^*_t$ has removed the pause between $a_t$
    and~$b_t$; hence,~$\Gamma_s \in \elems\ipaths$.
    Putting these maps together, we obtain a
    map~$\Gamma^* \colon [0,1] \to \elems\cpaths$, where
    \begin{equation}\label{eq:reparam-2}
        \Gamma(t) :=
        \begin{cases}
            \Gamma_t &\text{if } t \not\in (t_1,t_2)\\
            \Gamma^*_t &\text{if } t \in (t_1,t_2).
        \end{cases}
    \end{equation}
    Moreover, $\Gamma$ is continuous in $\ipaths$.
\end{proof}

Direct linear interpolation can also yield degeneracies by creating a singleton in specific circumstances, or by creating a backtracking point.
Each are addressed in the following theorem, and a path is constructed.

\begin{theorem}[Path Immersions]\label{thm:path-imm}
    The extended metric space~$(\ipaths, d_{FP})$ of paths immersed in $\R^n$ is
    path-connected iff $n>1$.
\end{theorem}

\begin{proof}
    If $n=1$, it is easy to see that $\ipaths$ is not path-connected by
    examining intervals with reversed orientation which trivially degenerate to
    a point when constructing a path, violating local injectivity.

    Now, consider $n>1$.  Let $\eqpath{0},\eqpath{1} \in \ipaths$.
    Using \defref{interp-graphs}, let~$\Gamma: [0,1] \to \cpaths$ be the linear interpolation from
    $\apath_0$ to $\apath_1$.
    This interpolation is in~$\cpaths$, not~$\ipaths$, so we explain how to
    edit~$\Gamma$ so that it stays in $\ipaths$.
    If~$\Gamma(t) \in \elems\ipaths$ for each $t \in [0,1]$, we
    are done.  Otherwise, let~$T \subset I$ be the set of times that introduce a
    non-immersion (i.e.,~$t \in T$ iff $\Gamma(t) \not\in \elems\ipaths$,
    but~$\Gamma(t-\epsilon) \in \elems\ipaths$ for all $\epsilon$ small enough).
    There are two things that might have happened at $t$:
    either an interval collapsed to a point (a pause) or backtracking was introduced
    in $\Gamma(t)$.

    \begin{enumerate}

        \item Suppose there exists $t \in T$ where an interval pauses as in \defref{pause}
            and \figref{point}. Note that a pausing event occurs either if an
            interval of ~$\Gamma_t$ becomes degenerate, or $\Gamma_t$
            collapses to a point.

            If pausing occurs only on an open interval $(a,b) \subset [0,1]$ of $\apath_t \in \Gamma_t$, it can be avoided using \lemref{thepause}.
            If pausing occurs on a closed interval $[a,b] \subset [0,1]$,
            we convert it to the open set $(a-\eps, b+\eps)$
            \journal{not sure if this fits assumptions of \lemref{thepause}}
            for small~$\eps$, and use \lemref{thepause}. If either $a=0$
            or~$b=1$, we simply redefine $\apath_t$ to
            start at $b$ or to end at $a$, respectively, using \lemref{closedpause}.
            The pausing event is guaranteed to conclude at some $t + \delta$ for $\delta \geq 0$
            since~$\eqpath{1} \in \ipaths$, and $\Gamma$ must attain $\apath_1 \in \eqpath{1}$.

            If a pausing event stems from a full collapse to a singleton (i.e. interpolation occurs between
            two colinear segments with reverse orientation, and consequently degenerate to a point),
            the collapse can be circumvented by rotating the path defining~$\Gamma_t$, which is done in
            \lemref{ole-spinny}.

        \item Alternatively, suppose there exists $t \in T$ which corresponds to backtracking at
            a point in a path~$\Gamma_t$ according to \defref{backtrack} and
            \figref{backtrack}. Here,~$\Gamma_t$
            can remain in~$\elems\ipaths$ by inflating a ball of radius~$\epsilon$ for sufficiently small $\epsilon > 0$
            about the backtracking point before it is created.
            This is included in \lemref{q-tip}, and shown in \figref{q-tip}.

    \end{enumerate}

    For all $t \in T$, the described moves can be used to subvert lapses in local injectivity
    along $\Gamma$. Hence, we construct a path $\Gamma$
    by interpolating from~$\apath_0$ to~$\apath_1$, and applying the required move at each
    $t \in T$ to handle pauses or backtracking. By the arbitrariness of $\Gamma$, we have given
    a class of continuous paths from any element~$\apath_0 \in \eqpath{0}$ to any
    $\apath_1 \in \eqpath{1}$.
\end{proof}

\begin{theorem}[Metric Balls in $(\ipaths, d_{FP})$]\label{thm:imm-path-balls}
    If $n>1$, then balls in the extended metric space $(\ipaths, d_{FP})$ are path-connected.
\end{theorem}

\begin{proof}
	Let $\eqpath{0}, \eqpath{1} \in \ipaths$, and let $\delta>0$.
        Let $\Gamma \in \ipaths$ be the map $\Gamma$ in the proof of
        \thmref{path-imm}.
        \journal{BTF: can we simplify this so we're not chasing references?}
	By \lemref{cont-graph-balls}, linear interpolation does not increase
        the Fr\'echet distance. By design, avoiding singleton degeneracies by way of \lemref{ole-spinny}
	also does not increase the Fr\'echet distance. Moreover, by
        construction, the map~$\Gamma^*$ of
	\lemref{thepause} preserves the Fr\'echet distance.
        \journal{the previous sentence needs a lemma}
    The maneuver in \lemref{q-tip} could potentially increase~$d_{FP}(\Gamma_t, \apath_1)$
	at some time $t \in [0,1]$,
	but in this case any critical backtracking points
    can be perturbed slightly in order to no longer define the $d_{FP}(\Gamma_t, \apath_1)$. Hence,
	these moves need not result in $d_{FP}(\Gamma_t, \apath_1) > \delta$, meaning that
    $\Gamma_t \in \ball{d_{FP}}{\agraph_1}{\delta}$,
	and balls in $\ipaths$ are path-connected.
\end{proof}

We use the same maneuvers from \thmref{path-imm} in the context for graphs under $d_{FG}$.

\begin{theorem}[Graph Immersions]\label{thm:graph-imm}
    For every graph~$G$, the extended metric space $(\igraphs(G), d_{FG})$ is path-connected.
    Connected components of the extended metric space $(\igraphs,d_{FG})$
    are in one-to-one correspondence with the homeomorphism classes of~graphs.
\end{theorem}

\begin{proof}
    We construct $\Gamma$ identically to \thmref{path-imm},
    but interpolation
    occurs among each edge of $G$ in $\igraphs(G)$
    rather than between individual segments. As in \thmref{path-imm},
    local injectivity can only be violated by pauses and backtracking on edges,
    which are handled using \lemref{thepause}, \lemref{ole-spinny}, and \lemref{q-tip} on each edge.
    If $(G,\eqgraph{0}), (H,\eqgraph{1}) \in \igraphs$ for $G, H$ which are not homeomorphic,
    then $d_{FG}((G,\eqgraph{0}), (H,\eqgraph{1})) = \infty$.
\end{proof}

Similarly, we can adopt \thmref{imm-path-balls}
for each edge in a graph to show path-connectivity of balls in $\igraphs$.

\begin{theorem}[Metric Balls in $(\igraphs(G), d_{FG})$]\label{thm:imm-graph-balls}
    For every graph~$G$, the balls in the extended metric space~$(\igraphs(G), d_{FG})$ are path-connected.
\end{theorem}

\begin{proof}
    Let $\eqgraph{} \in \igraphs$, and let $\delta > 0$.
    Let~$\B$ be the intersection~$\ball{d_{FG}}{\agraph}{\delta}\cap \igraphs(G)$.
    Let $\eqgraph{0}, \eqgraph{1} \in \B$.
    Construct the path~$\Gamma:[0,1]\to \B$ from $\eqgraph{0}$ to $\eqgraph{1}$
    in the same way as \thmref{graph-imm}.
    \journal{BTF: can we simplify this so we're not chasing references?}
    Just as in \thmref{imm-path-balls},
    linear interpolation
    and the moves in \lemref{ole-spinny}, \lemref{q-tip}, and \lemref{thepause}
    mandate that~$\Gamma(t) \in \B$ for every $t \in (0,1)$ identically to the
    path Fr\'echet distance.
\end{proof}

\subsection{Embeddings}\label{sec:emb}
Lastly, we examine the path-connectedness property of the analogous spaces of embeddings.

\begin{theorem}[Path Embeddings]\label{thm:path-emb}
    The extended metric space $(\epaths, d_{FP})$ is path-connected in $\R^n$
    if and only if~$n > 1$.
\end{theorem}

\begin{proof}
    If $n=1$, two paths with reverse orientations are not path-connected.

    Now, let $n >1$, and
    let $\eqpath{0},\eqpath{1} \in \epaths$.
    \journal{the following is not very precise}
    By Alexander's trick,\footnote{Two embeddings of the $n$-ball are isotopic,
    first proven by Alexander~\cite{alexander1923deformation}; see
    also~\cite[\textsection 4]{denne2008convergence}.}
    there exists~$s_0 \in [0,1]$ such that~$\apath_0' := \apath_0 |_{[s,1]}$
    and $s_1 \in [0,1]$ such that $\apath_1' := \apath_1
    |_{[s_1,1]}$, where~$s_0$ and~$s_1$ are nearly straight. Let
    $\angle$ be the angle between the segments $\apath_0'$
    and $\apath_1'$.
    Let $S \colon [\frac{1}{4},\frac{2}{4}] \to \epaths$ be the map rotating $\apath_0'$ by $\angle$ to become
    parallel with $\apath_1'$.
    Finally, let $\Gamma$ be the interpolation from $\apath_0' \circ S$ to $\apath_1'$.

    Define $P \colon [0,1] \to \epaths$ as the resulting composition:
    \begin{equation*}
        P(t) =
            \begin{cases}
                \apath_0 ~|_{[(1-t)s,1]}, & t \in [0,\frac{1}{4}] \\
                S(t), & t \in [\frac{1}{4},\frac{2}{4}] \\
                \Gamma(t) & t \in [\frac{2}{4},\frac{3}{4}] \\
                \apath_1 ~|_{[(1-t)s,1]}, & t \in [\frac{3}{4},1].
            \end{cases}
    \end{equation*}

    The steps attaining $\apath_0'$ and $\apath_1'$, as nothing else than a restriction
    of $\apath_0$ and $\apath_1$, are continuous. Moreover, $S$ is continuous as the rotation of
    $\apath_0'$, and $\Gamma$ is continuous by \lemref{gamma-cont}.
    By the arbitrariness of the constructed path and $\apath_0, \apath_1$,
    there is a family of continuous paths for any $\apath_0 \in \eqpath{0}, \apath_1 \in \eqpath{1}$,
    and $\epaths$ is path-connected.
\end{proof}

\begin{figure}
  \centering
  \begin{minipage}[b]{0.2\textwidth}
    \includegraphics[width=\textwidth]{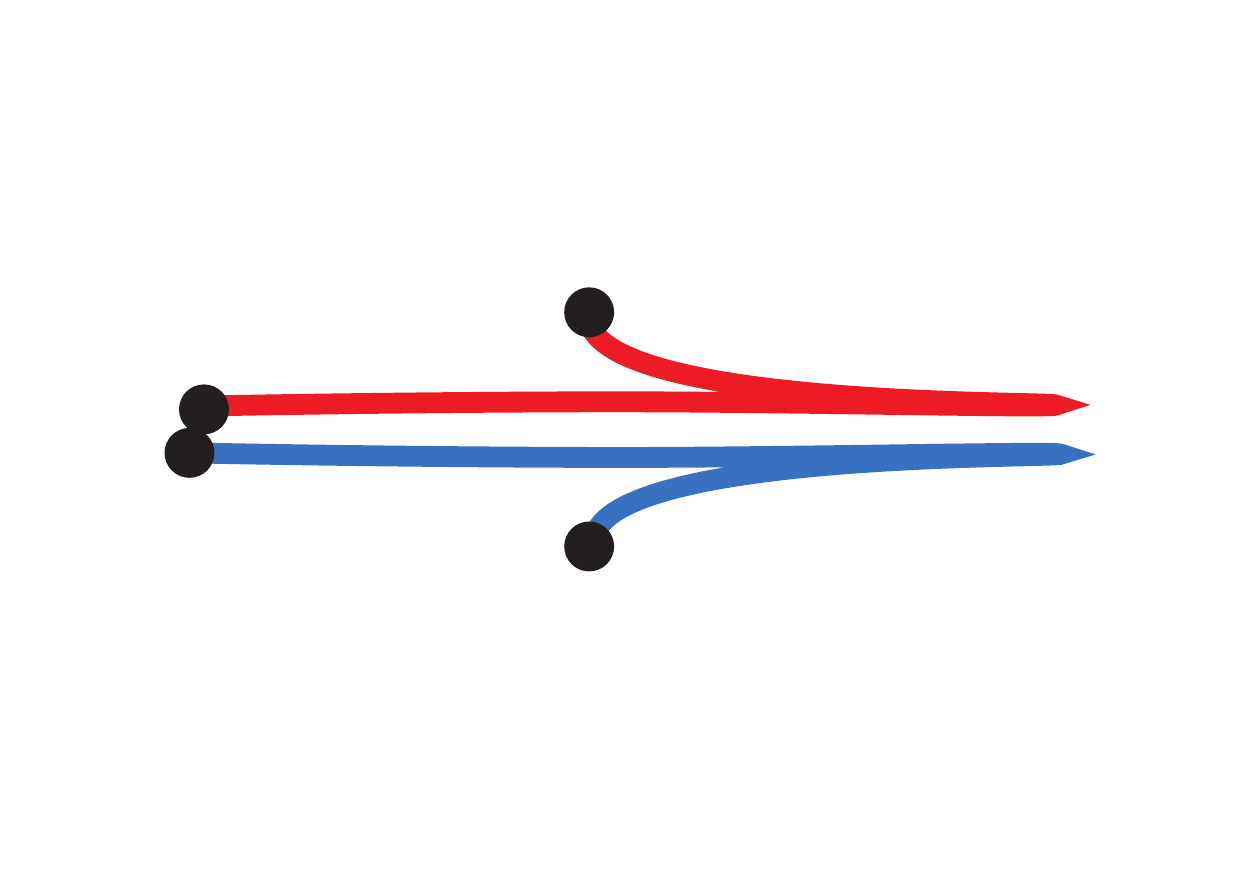}
  \end{minipage}
        \hspace{1cm}
  \begin{minipage}[b]{0.2\textwidth}
    \includegraphics[width=\textwidth]{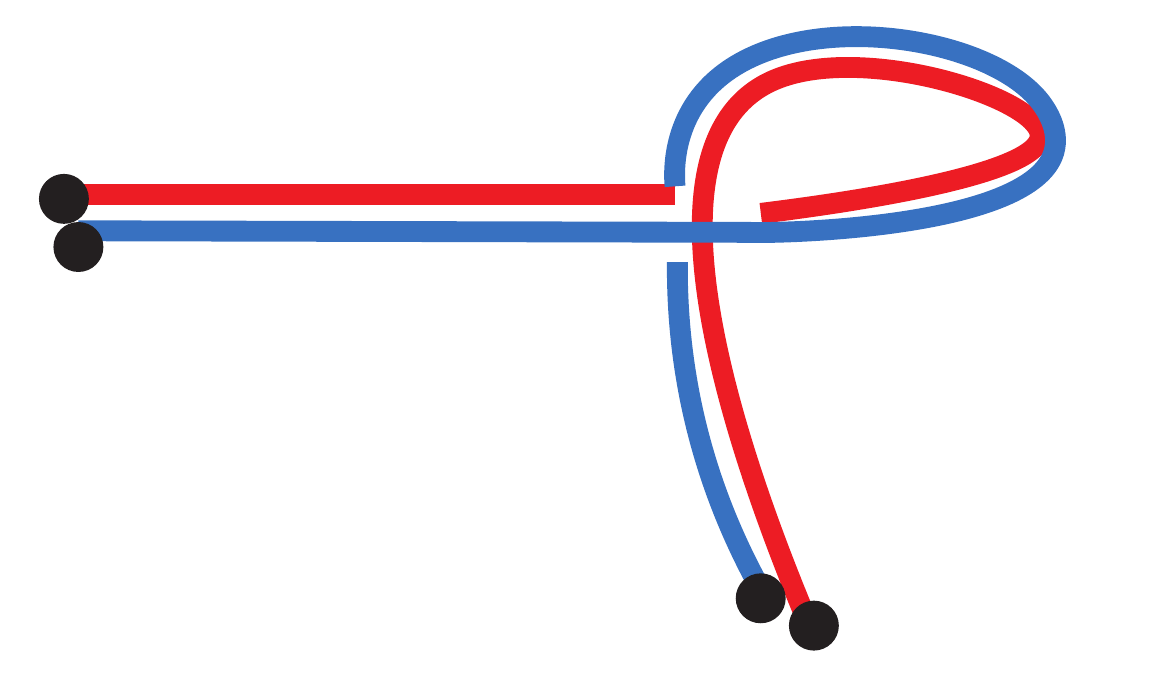}
  \end{minipage}
   \caption{Two embedded paths $\apath_0, \apath_1$ in $\R^2$ and $\R^3$ respectively, for which constructing a path
	$\Gamma:[0,1] \to \epaths, \Gamma(0) = \apath_0, \Gamma(1) = \apath_1$ is not possible
	without having~$\Gamma(t) \not \in \ball{d_{FP}}{\apath_1}{d_{FP}(\apath_0, \apath_1)}$ for some
	$t \in [0, 1]$.} \label{fig:dimembed}
\end{figure}

Moreover, in high dimensions we can construct a path in $\epaths$ not increasing the Fr\'echet distance.

\begin{theorem}[Metric Balls in $(\epaths, d_{FP})$]
    If $n \geq 4$, then balls with finite radius in the extended metric space~$(\epaths, d_{FP})$ are path-connected
    in~$\R^n$.
\end{theorem}

\begin{proof}
    If $n \geq 4$, the same map $\Gamma$ given in \thmref{path-imm}
    is sufficient, except that self-crossings must be avoided.
    At each $s \in [0,1]$ where a self-crossing would occur,
    we perturb $\Gamma$ by a sufficiently small amount in order to avoid a self-crossing
    without increasing the Fr\'echet distance using the maneuver in \lemref{embed-path-balls}.
\end{proof}

A simple examination shows that metric balls are not path-connected in low dimensions.

\begin{theorem}[Metric Balls in $(\epaths, d_{FP})$]
    If $n \in \{1,2,3\}$, then balls with finite radius in the extended metric space~$(\epaths, d_{FP})$ are not path-connected
    in $\R^n$.
\end{theorem}
\begin{proof}
    For $n=1$, let $\eqpath{} \in \epaths$. Let $\apath^{-1}:=\apath(1-t)$, and
    note that $\apath^{-1} \in \epaths$. If $n=2$, consider two paths within a fixed Fr\'echet ball
    that are much wider than their Fr\'echet distance.
    If $n=3$, consider two paths with small Fr\'echet distance that form a loop, with one section passing under
    the other. If these loops have reversed orientation between the two paths, the Fr\'echet distance must increase.
    See \figref{dimembed} for examples.
\end{proof}

In the setting for graphs, the path-connectedness property reduces to a knot theory problem if $n \leq 3$,
and is not maintained.
For $n \geq 4$, we use the existence of a sequence of Reidemeister
moves~\cite{reidemeister1927knoten}
\journal{citation and/or appendix for R-moves is needed (per convo w/ Brittany
and Erin
on 6/23)}
from any tame knot to another to construct paths in $\egraphs$.

\begin{theorem}[Path-Connectivity of $(\egraphs, d_{FG})$, $n \geq 4$]\label{thm:graphs-high}
    For all graphs $G$ and $n \geq 4$, the extended metric space~$(\egraphs(G),d_{FG})$ is path-connected.
    Moreover, connected components of the
    extended metric space~$(\egraphs,d_{FG})$ are in one-to-one correspondence with
    homeomorphism classes of graphs.
\end{theorem}

\begin{proof}
        Let $G$ be a graph, and $\agraph_0, \agraph_1 \in \egraphs(G)$.
        If~$n\geq4$, any tame knot can be unwound by a sequence of
        Reidemeister moves into the unknot. Construct
        $\Gamma:[0,1]\to\egraphs(G)$ by linear interpolating
        until some~$t \in (0,1)$ causes $\Gamma_t$ to self-intersect. At~$t$, there
        exists a Reidemeister move allowing the crossing event to occur.
        Hence, any sequence of knots and free
        edges comprising $\phi_0$ and $\phi_1$ can be unwound to a sequence of unknots and straight edges, and
        then interpolated accordingly. Consequently, there exists a
        path from~$\agraph_0$
        to~$\agraph_1$ in~$(\egraphs(G), d_{FG})$. Note that we require~$\phi_0, \phi_1$
        are rectifiable in \secref{background}.
    \journal{BTF: this isn't a part of a proof, but might make a good remark
    somewher:
        Without this
        requirement,~$\agraph_0$ and~$\agraph_1$ could comprise wild knots, and constructing such a path could
        have infinitely many Reidemeister~moves.
    }
\end{proof}

In dimension $4$ or higher, the path-connectivity of balls in $\egraphs(G)$
is shown in the same way as for paths.

\begin{theorem}[Metric Balls in $(\egraphs(G), d_{FG})$]
    For all graphs $G$ and $n \geq 4$, metric balls in the space~$(\egraphs(G),
    d_{FG})$ are path-connected.
\end{theorem}

\begin{proof}
The proof is identical to that in \lemref{embed-path-balls}, but Reidemeister
moves are used for each edge in a graph rather than a single segment.
\end{proof}

\section{Conclusion}\label{sec:conclusion}
In this paper, we studied some fundamental topological properties of 
spaces of paths and graphs in Euclidean space under the Fr\'echet distance. 
In particular, we investigated  metric properties of the Fr\'echet distance 
on paths and graphs, as well as studying the path-connectedness of metric balls in the space of
such graphs.  While this work is theoretical and mathematical in nature, we feel that 
establishing the underlying properties of the topological spaces it
can define provides an important theoretical backdrop, which is especially critical 
due to the widespread popularity of the Fr\'echet distance in computational
geometry, and the growing popularity of its extension for graphs. 
Our contribution begins a careful study of the Fr\'echet distance and its topological 
properties. Extensions to this work abound, and include examining core topological 
properties of other distance measures in computational geometry, as well as
other important properties of the Fr\'echet~distance.

\section*{Acknowledgements}
We thank our coauthors from~\cite{rogue-1} for the initial conversations that
inspired this paper. In addition, the authors thank the National Science
Foundation; specifically, Erin Chambers is partially supported by NSF awards
2106672 and 1907612, Brittany Terese Fasy and Benjamin Holmgren are partially
supported by award 1664858, Fasy is also supported by award 2046730, and Carola
Wenk is partially supported by award 2107434.

\bibliographystyle{abbrv}
\bibliography{references}

\newpage
\appendix

\section{Distances and Topology}\label{append:basics}
Let $\bar{\R}$ denote the extended real line: $\bar{\R}=\R \cup \pm \infty$.
We now provide the basic definitions relating to distances and topology used
throughout this paper.

\subsection{Graphs}\label{append:graphs}
Graphs are a central object studied in this paper.

Throughout this paper, we use the term \emph{graph} to mean a
\emph{multi-graph}.
A multi-graph $G=(V,E)$ is a finite set of vertices $V$ and edges $E$.
Self-loops and multiple edges between two vertices are allowed in this setting.
A graph is an example of a more general structure called a CW complex, which we
topologize as follows: (1) the topology on $G$ restricted to $V$ is the discrete topology;
(2)  for a edge $e$, the open sets restricted to is closure $\bar{e}$ are those induced by
the subspace topology on~$[0,1] \subset \R$ and a homeomorphism $[0,1]\to
\bar{e}$; (3) we take
the quotient topology on $\left( \cup_{v \in V} \right) \cup \left( \cup_{e\in E} \bar{e}
\right)$.

\subsection{Fr\'echet Distance}\label{append:frechet}

We defined the path and graph Fr\'echet distances in \secref{background}.  The
path Fr\'echet distance is well-studied~\cite{Frechet, Fang2021, buchin2010frechet, bos-cfdrvs-17, Buchin2020, Alt1995, Aronov2006,Agarwal2014,Driemel2012,Alt2001,Driemel2013,Cook2010,CHAMBERS2010295,Driemel2022,mapmatching,Colombe2021,Guibas2011}.   The graph
Fr\'echet distance has been less studied, but many results for paths transfer to
graphs.

\journal{give example where infimum is not attained}

The proof of the following lemma follows from the definition of Fr\'echet
distance and the definition of infimum.
\begin{lemma}[Approximator]\label{lem:frapprox}\label{lem:hexist}
    For all graphs~$G$,
    if~$\eqgraph{0},\eqgraph{1} \in \cpaths(G)$, then
    for every $\eps>0$, there exists a homeomorphism $h_* \colon G \to G$ such
    that
    $$
        ||\agraph_0 - \agraph_1 \circ h ||_{\infty} < d_{FG}(\agraph_0,\agraph_1) + \epsilon.
    $$
\end{lemma}

\begin{proof}
    By \defref{graphfrech},
    $$d_{FG}(\eqgraph{1}, \eqgraph{2}) = \inf_h||\agraph_1 - \agraph_2 \circ
    h||_{\infty}.$$
    Then, by the definition of infimum, for every $\eps > 0$, there exists
    $h_*:G \to G$ such that
    \begin{align*}
        ||\agraph_1 - \agraph_2 \circ h_*||_{\infty}&
        < \inf_h||\agraph_1 - \agraph_2 \circ h||_{\infty} + \eps/2 \\
        & = d_{FG}(\eqgraph{1}, \eqgraph{2}) + \eps/2,
    \end{align*}
    as was to be shown.
\end{proof}

\journal{relationship to Hausdorff}

\subsection{Defining Spaces from Distances and Metrics}\label{append:distdefs}

Given a set $\X$ and a $\comparefcn{d}{\X}$, we
topologize $\X$ as follows:

\begin{definition}[The Open Ball Topology]\label{def:open-ball}
    Let~$\X$ be a set and $\comparefcn{d}{\X}$ a distance function.  For each~$r > 0$ and~$x \in \X$,
    let~$\ball{d}{x}{r} := \{ y \in \X ~|~ d(x,y) < r \}$. The open ball
    topology on $\X$ with respect to $d$ is the topology generated
    by~$\{ \ball{d}{x}{r} ~|~ x \in \X, r>0\}$.
    We call~$(\X,d)$ a \emph{distance space}.
\end{definition}

In words, $\ball{d}{x}{r}$ denotes the open ball of radius $r$ centered
at $x$ with respect to $d$.
We use these open balls to generate a topology on $\X$,
allowing~$x$ to range over $\X$ and $r$ to range over
all positive real numbers.

We are particularly interested in distance functions that are either a
pseudo-metric or a metric. These are defined as follows.

\begin{definition}[Pseudo-Metric]\label{def:pseudometric}
    Let~$\X$ be a set and let~$\comparefcn{d}{\X}$ be a distance function.
    We call $d$ a \emph{pseudo-metric} on $\X$ if $d$ satisfies the following:
    \begin{itemize}
        \item Finiteness: $d(x,y) < \infty$ for all $x,y \in \X$.
        \item Identity: $d(x,x)=0$ for all $x \in \X$.
        \item Symmetry: $d(x,y)=d(y,x)$ for all $x,y \in \X$.
        \item Subadditivity (the triangle inequality): $d(x,z)\leq d(x,y) + d(y,z)$ for all $x,y,z \in \X$
    \end{itemize}
    If $d$ satisfies everything except finiteness, then we call~$d$ an
    \emph{extended pseudo-metric}.
\end{definition}

In order to be a metric, $d$ must fulfill stricter criteria:

\begin{definition}[Metric]
    Let~$\X$ be a set and let the function~$\comparefcn{d}{\X}$ be a pseudo-metric.
    We say that $d$ is a \emph{metric} if $d$ also
    satisfies:
    \begin{itemize}
        \item Seperability: for any $x,y \in \X$, if $x\not=y$, then~$d(x,y) >0$.
    \end{itemize}
\end{definition}

Often, if $(\X,d)$ is a pseudo-metric space, a standard procedure is to define an equivalence class for $x,y \in \X$ where
$x \sim y$ if $d(x,y) = 0$. Then, the quotient space~$\bigslant{\X}{\sim}$ is a metric space.

Common examples of metrics on function spaces are those induced by $L_p$-norms.
For example, let $(\Y,d_{\Y})$ be
distance space, let $\X$ be any topological space, and let~$f,g \colon \X \to \Y$.
Then, the distance induced by the~$L_{\infty}$-norm between $f$ and $g$ is:
$$||f-g||_{\infty} = \max_{x \in \X} d_{\Y}(f(x),f(y)).$$

\subsection{Paths and Maps}\label{append:path}

With the basic definitions from topology in hand, we are equipped to define a property
of fundamental interest in topology: path-connectedness.

\begin{definition}[Path]\label{def:path}
    A \emph{path} in a topological space $\X$ between two elements $a,b \in \X$,
    is defined to be a continuous map $\gamma:[0,1] \to \X$ where $\gamma(0)=a$, and
    $\gamma(1) = b$.
\end{definition}

Given two paths $\apath_1,\apath_2 \colon [0,1] \to \X$ such that
$\apath_1(1)=\apath_2(0)$, we combine them by taking both at double-speed.
This is called the concatenation of paths. In particular, $\apath_1 \# \apath_2
\colon [0,1] \to \R^n$ is defined by:
$$
    \apath_1 \# \apath_2 (t) :=
        \begin{cases}
            \apath_1(2t) & t \in [0,0.5].\\
            \apath_2(2t-1) & \text{otherwise}.
        \end{cases}
$$
Given one path $\apath \colon [0,1] \to \X$ and an interval $[a,b] \subseteq
[0,1]$, the restriction of $\apath$ to $[a,b]$ is also a path, given~by:
$$
    \apath |_{[a,b]}(t) := \apath(a + t(b-a)).
$$

With the definition of paths, we define a primary property of interest in this paper:
path-connectivity.

\begin{definition}[Path-Connectivity]\label{def:path-con}
    A topological space $\X$ is called \emph{path-connected} if there exists a \emph{path}
    between any two elements in $\X$.
\end{definition}

We also define the path-connectedness property specifically for distance balls:

\begin{definition}[Path-Connectivity of Balls]\label{def:ball-con}
    Let~$(\X,d)$ be a topological space, let $x \in \X$.
    We say that the distance balls in~$(\X,d)$ are \emph{path-connected} if for
    every $x \in \X$ and~$r \in \R_{\geq 0}$, the distance ball~$\ball{d}{y}{r}$ is path-connected.
\end{definition}

And, the \emph{length} of a path in a distance space
is given~by:

\begin{definition}[Length]\label{def:length}
    Let~$(\X, d)$ be a distance space and let $\gamma$ be a \emph{path} in $(\X, d)$.
    Let~$\mathcal{P}$ be the set of all finite subsets $P= \{t_i\}$  of $[0,1]$ such that
    such that~$0 = t_0 < t_1 < \ldots < t_n = 1$. The
    \emph{length} $L_d(\gamma)$
    of~$\gamma$ is:
    \[ L_d(\gamma) := \sup_{P \in \mathcal{P}} \sum_{i=1}^{n}{d(\gamma(t_i), \gamma(x_{i-1}))}. \]
\end{definition}

Additionally, it is often useful in our setting to reparameterize paths, both
to define the Fr\'echet distance and to maintain properties such as injectivity in a map.

\begin{definition}[Reparameterization]
    Let $\X,\Y$ be a topological spaces, $\phi \colon \X \to \Y$, and $h \colon
    \X \to \X$ is a homeomorphism. Then, we call $\phi \circ h$ a
    reparameterization of~$\phi$.
    In the setting where $\X=[0,1]$ and $h(0)=0$, we call $\phi \circ h$ an
    \emph{orientation-preserving} reparameterization.
\end{definition}

\section{Omitted Details for Path-Connectivity}
In this appendix, we provide additional context for the proofs of path-connectivity in
\secref{path-conn-sec}.

\subsection{Additional Details on Interpolation}\label{append:lin-interp}

Given two continuous maps of the same graph into~$\R^n$, we define the
interpolation between them.  First, we need to define linear combinations of
graphs (and paths).

\begin{definition}[Linear Combination of Graphs]\label{def:combo}
    Let $G$ be a graph, let $\phi_1,\phi_2:[0,1]\to \R^n$ be continuous,
    rectifiable maps,
    and~$c_1,c_2
    \in \R$. Then, the linear combination~$\phi= c_1 \phi_1 + c_2\phi_2$
    is defined as follows: the map~$\phi \colon G \to \R^n$ is defined
    by~$\phi(x) := c_2
    \phi_1(x) + c_2 \phi_2(x)$.  In this case, we may also say $(G,\phi)$ is a
    linear combination of graph-map pairs $(G,\phi_1)$ and~$(G,\phi_2)$.
\end{definition}

In this definition, we observe that $\phi$ is continuous (since $\phi_0$ and~$\phi_1$ are
continuous), which means that linear combinations are well-defined in the set of all continuous,
rectifiable maps. It is not well defined in the space~$\cgraphs(G)$ overall.
\journal{BTF: this is well-defined in the set of all maps, but not in the equiv.
classes. maybe have an example to demonstrate?}

\begin{lemma}[Linear Interpolation is Continuous]\label{lem:gamma-cont}
    For all graphs~$G$,
    linear interpolation between graphs in~$\cgraphs(G)$
    (and hence between homeomorphic graphs in~$\cgraphs$)
    is a continuous function.
\end{lemma}

\begin{proof}
    Let $\eqgraph{0}, \eqgraph{1} \in \cgraphs(G)$. Let $\Gamma \colon [0,1] \to
    \cgraphs(G)$ be the linear interpolation from $\agraph_0$ to $\agraph_1$.

    We prove that $\Gamma$ satisfies the $\eps$-$\delta$ definition of
    continuity.
    Let~$\eps > 0$.
    Set $\delta = \frac{\eps}{d_{FG}(\eqgraph{0}, \eqgraph{1})}$.
    Let $s,t \in [0,1]$ such that $|s-t|< \delta$.
    Then, we have
    \begin{align*}
        & d_{FG} \left( [\Gamma_t], [\Gamma_s] \right) \\
        &~=d_{FG}([(1-t)\agraph_0 + t\agraph_1], [(1-s)\agraph_0 + s\agraph_1]) \\
        &~=\inf_h~|| ((1-t)\agraph_0 + t\agraph_1) - ((1-s)\agraph_0
            + s\agraph_1)\circ h||_{\infty},
    \end{align*}
    where $h$ ranges over all reparameterizations of $[0,1]$.
    Continuing, we find:
    \begin{align*}
        & d_{FG} \left( [\Gamma_t], [\Gamma_s] \right) \\
        &~=\inf_h~|| (s-t)\agraph_0 + (t-s)\agraph_1 \circ h||_{\infty} \\
        &~=|t-s| \inf_h ~|| -\agraph_0 + \agraph_1 \circ h||_{\infty} \\
        &~<\delta \cdot \inf_h~|| -\agraph_0 + \agraph_1 \circ h||_{\infty} \\
        &~=\eps,
    \end{align*}
    by definition of $\delta$.

    And so, we have shown that $\Gamma$ satisfies the
    the $\eps$-$\delta$ definition of continuity.
    In extended metric spaces, the $\eps$-$\delta$ definition of continuity is
    equivalent to topological continuity
    (e.g., see proof
    in~\cite[Lemma 7.5.7]{Lebl} for metric spaces).
    Thus, we conclude that linear interpolation between
    graphs in $\cgraphs{G}$ is continuous.
\end{proof}

Setting~$G=[0,1]$, an identical argument shows that linear interpolation
between paths in $\cpaths$ is also~continuous.

\begin{cor}[Linear Interpolation between Paths]
    For all $\eqpath{0}, \eqpath{1} \in \cpaths$, the linear interpolation
    from $\apath_0$ to~$\apath_1$ is continuous.
\end{cor}

\begin{figure}
    \centering
    \includegraphics[width=.3\textwidth]{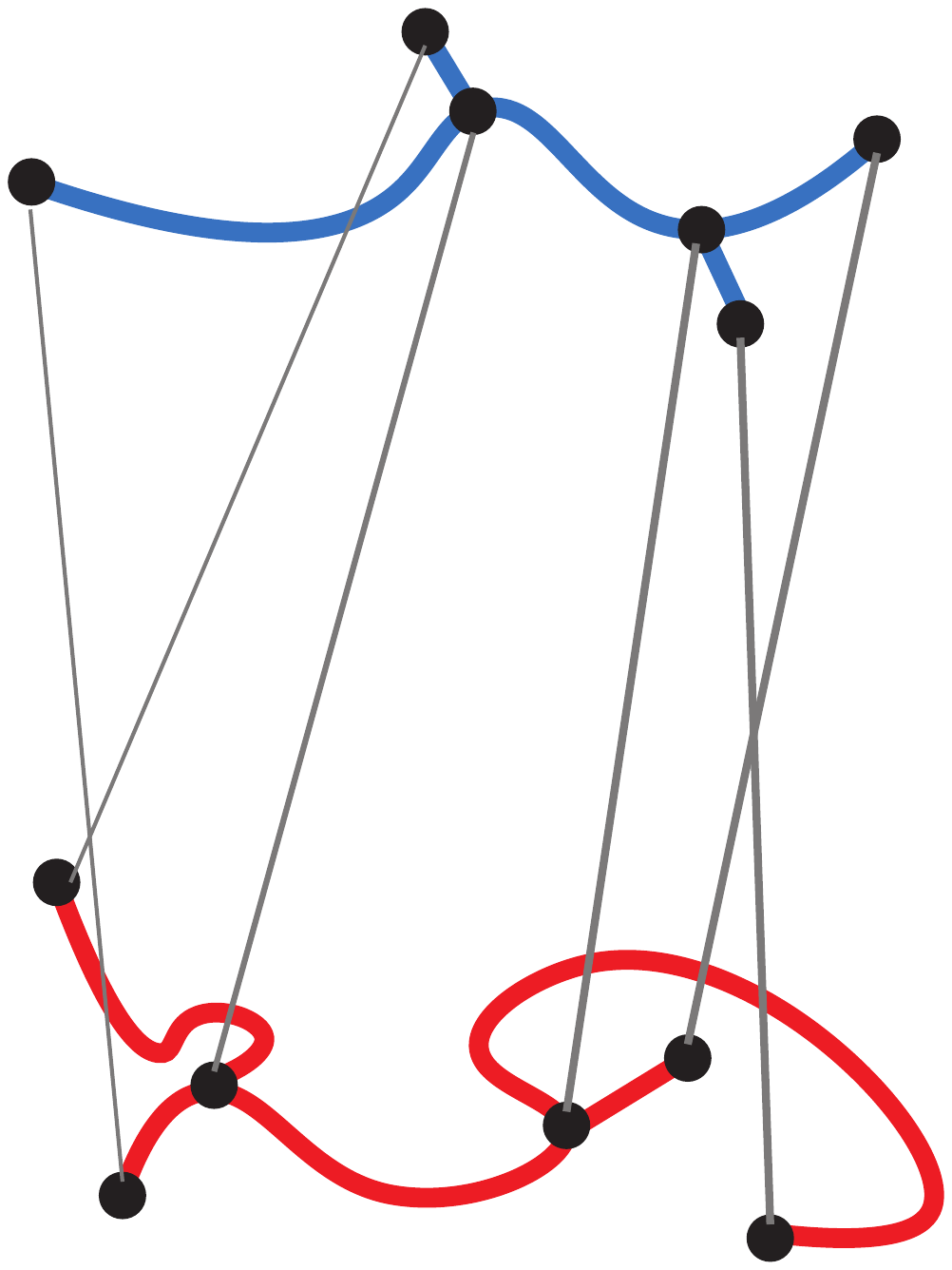}
    \caption{
    The interpolation between two embeddings of a graph in $\R^n$.
    For simplicity, we show the interpolation between the vertices in the
    embeddings,
    and the interpolation between edges is inferred accordingly.
    }\label{fig:interp}
\end{figure}

\subsection{Paths Between Immersions in Greater Detail}

This subsection includes additional details to maintain local injectivity for
an arbitrary path $\Gamma:[0,1]\to\ipaths$.
We begin by examining the case where pausing occurs on the closed interval
$[a,b]\subset[0,1]$ in the domain of an immersed path $\apath_t \in \Gamma_t$,
and the interval includes either~$0$ or~$1$. We subvert this by finding an
alternate (but `close') path $\Gamma^*$.

\begin{lemma}[Pausing at Endpoints]\label{lem:closedpause}
    Let $\eqpath{0}, \eqpath{1} \in \ipaths$, and let $\Gamma:[0,1]\to\cpaths$
    be a path in $\cpaths$ starting at~$\apath_0$
    and ending at $\apath_1$.
    Suppose that there exists an interval~$[t_1,t_2] \subset [0,1]$ such that
    for~$t
    \in [0,1]\setminus (t_1,t_2)$, $\Gamma_t$ is an immersion and, for~$t \in (0,1)$, $\Gamma_t$
    has a single pause (and no other violations of local injectivity).
    Then, there exists an alternate path~$\Gamma^*$ in $\cpaths$ starting at
    $\apath_0$ and ending at~$\apath_1$ such
    that~$\Gamma^*$ is a path in~$\ipaths$.
\end{lemma}
\begin{proof}
    \journal{check this proof for `can' and `it is'}
    We use the same idea as in \lemref{thepause}, but instead stretch the unit interval
    into only one side of the original domain of the path $\Gamma_t$.
    That is:
    \begin{equation}\label{eq:reparam-closed-1}
        \Gamma_t^*(x) :=
        \begin{cases}
            \Gamma_t(x\cdot a) & \text{if } b=1\\
            \Gamma_t((x-b)\cdot(1-b) + b) & \text{if } a=0
        \end{cases}
    \end{equation}

    If $\Gamma_t$ pauses on $[a,1]$, we know that the we can focus on the image
    of~$[0,a]$.
    \journal{previous sentence needs to be more precise}
    Likewise, if~$\Gamma_t$ pauses on $[0,b]$,
    we turn to the image of~$[b,1]$.
    Then, replace $\Gamma_t$ that pauses with the newly defined $\Gamma^*_t$.
    And so, we define a new map~$\Gamma^* \colon [0,1] \to \cpaths$ as follows:
    \begin{equation}\label{eq:reparam-closed-2}
        \Gamma^*(t) :=
        \begin{cases}
            \Gamma_t & \text{if } t \not\in (t-\eps, t+\delta)\\
            \Gamma^*_t &\text{if } t \in (t-\eps, t+\delta)
        \end{cases}
    \end{equation}

    Indeed, it is easy to verify that each $\Gamma_t^*$ preserves local injectivity so
    $\Gamma_t^* \in \ipaths$.
    Moreover,~$\Gamma^*$ is continuous.
\end{proof}

We now examine the case when linear interpolation
results in a singleton, which causes a degeneracy in spaces of immersions.
We give a maneuver to subvert this for paths.

\begin{lemma}[Dodging Singletons]\label{lem:ole-spinny}
    Let $\eqpath{0}, \eqpath{1} \in \ipaths$, and let $\Gamma:[0,1]\to\cpaths$ be a linear interpolation from $\apath_0$ to $\apath_1$.
    Let $t \in [0,1]$
    such that $\Gamma(t)$ is a constant map, forcing $\Gamma(t)\not \in \ipaths$. We can avoid this total degeneracy
    by rotating $\Gamma(t)$.
\end{lemma}

    \begin{proof}
        Linear interpolation of $\apath_0$ to $\apath_1$
        produces a singleton if the two equivalence classes of paths are colinear with reversed orientation.
        Hence, if $\Gamma_t$ degenerates to a constant map, there exists sufficiently small
        $\epsilon > 0$ to continuously rotate $\Gamma(t- \epsilon)$ by $\pi$
        without forcing $d_{FP}(\Gamma(t), \apath_1)>d_{FP}(\apath_0, \apath_1)$. Thereby reversing the
        orientation of $\Gamma(t + \epsilon)$, and avoiding the constant map for any $\apath_t \in \Gamma_t$.
        See \figref{spin} for an example.
    \end{proof}

    \begin{figure}
        \centering
        \begin{subfigure}[b]{0.3 \textwidth}
            \includegraphics[width=\textwidth, height=6cm]{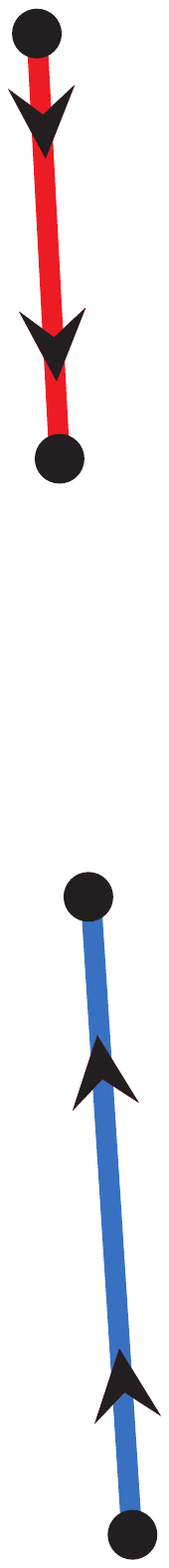}
            \caption{Paths with reversed orientation}
            \label{fig:reverse}
        \end{subfigure}
        \begin{subfigure}[b]{0.3 \textwidth}
        \includegraphics[width=\textwidth, height=6cm]{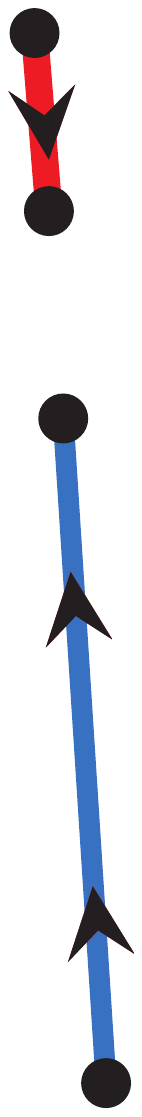}
        \caption{Interpolate}
        \label{fig:shrink}
        \end{subfigure}
        \begin{subfigure}[b]{0.3\textwidth}
            \includegraphics[width=\textwidth, height=6cm]{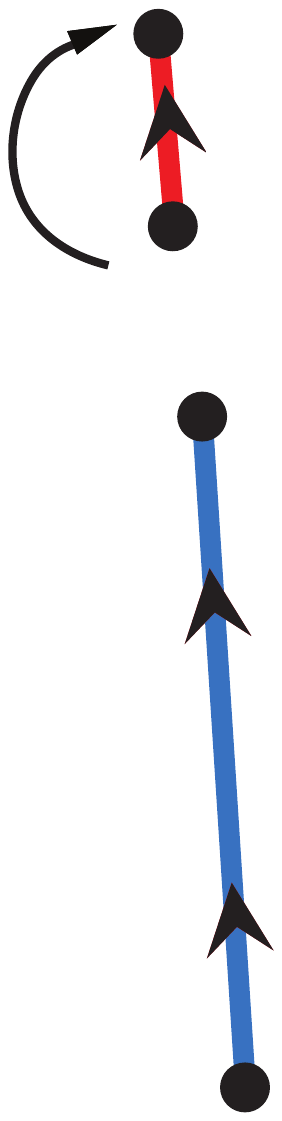}
            \caption{Rotate when sufficiently close}
            \label{fig:spin-180}
        \end{subfigure}
        \caption{
        For colinear paths with opposing orientation, rotating by $\pi$ avoids degenerating to
        the constant map, keeping $\Gamma$ in $\ipaths$. Moreover, rotation with sufficiently small
        Fr\'echet distance maintains the path-connectivity of balls.
        }\label{fig:spin}
    \end{figure}

We now consider the case of backtracking during linear interpolation, which violates local injectivity.
We introduce a maneuver to solve this potential degeneracy in spaces of immersions.

\begin{lemma}[The Q-Tip Maneuver]\label{lem:q-tip}
    Let $\eqpath{0}, \eqpath{1} \in \ipaths$, and let $\Gamma:[0,1]\to\cpaths$ be a linear interpolation from $\apath_0$ to $\apath_1$.
    Let $t \in [0,1]$ such that $\Gamma(t)$ creates backtracking for some $\Gamma(t)$.
    Inflating a ball about the critical backtracking point corrects this
    violation of injectivity.
\end{lemma}

\begin{proof}
In the scenario of a backtracking event, local injectivity is only violated at the exact critical
point $\Gamma_t(x)$ for $x \in [0,1]$ where backtracking occurs. For sufficiently small $\eps,\delta > 0$,
continuously inflate
a ball of radius $\delta$
about $\Gamma_{t-\eps}(x)$ such that $d_{FP}([\Gamma{t- \epsilon}], \eqpath{1})$ remains fixed, creating
the path $\Gamma^*_t$ with a ball replacing the critical point, so that $\Gamma^*_t \in \elems\ipaths$.
Then, replace any backtracking $\Gamma_t$ with the corresponding $\Gamma^*_t$.
For every $t \in [0,1]$ it holds that $\Gamma_t \in \ipaths$, and by the continuity of the inflation,
$\Gamma$ remains continuous.
For an example of this maneuver, see \figref{q-tip}.
\end{proof}

\begin{figure}
    \centering
    \begin{subfigure}[b]{0.4 \textwidth}
        \centering
        \includegraphics[height=1in]{bad-1}
        \caption{Example path with backtracking}
        \label{fig:backtracking}
    \end{subfigure}
    \begin{subfigure}[b]{0.4 \textwidth}
        \centering
	\includegraphics[height=1in]{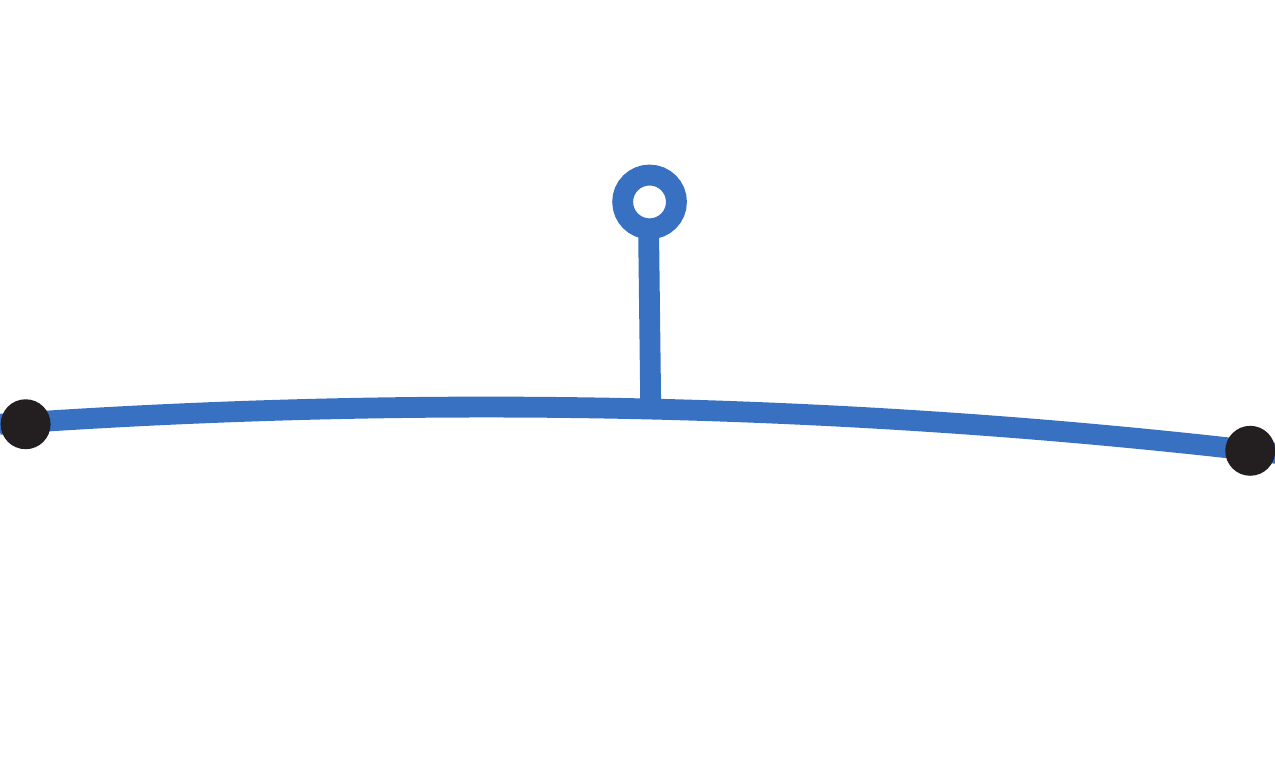}
	\caption{Inflate the critical backtracking point}
	\label{fig:q-tip}
    \end{subfigure}
    \caption{
    Reconcile forced backtracking along a path by inflating a ball about the critical
    backtracking point, thereby maintaining local injectivity.
    }\label{fig:q-tip-full}
\end{figure}

\subsection{Balls of Path Embeddings in Greater Detail}

In what follows, we elaborate on counterexamples for the path-connectivity
of balls in $\epaths$ and $\egraphs$. We begin with a counterexample for path embeddings
in $\R^2$.

We continue with a brief description of counterexamples for the path-connectivity of
embedded paths in $\R^3$.

\begin{lemma}[3d Balls in $\epaths$]\label{lem:balls-bad-dim3}
    If $n=3$, metric balls in the space $(\epaths, d_{FP})$ of embedded paths in
    $\R^n$ are not path-connected.
\end{lemma}

\begin{proof}
Metric balls are not in general path-connected in three dimensions. For a simple counterexample,
suppose $\apath_0$ comprises a loop in $\R^3$, where a segment crossed on top of itself, avoiding
self-intersection by some small distance $\delta$, with long tails at either
end of the crossing of length $2\delta$. Suppose also that $\apath_1$ comprises the mirror image of
$\apath_0$. Then, $d_{FP} = \delta$, but it is not possible to construct a path from $\apath_0$ to
$\apath_1$ without increasing the Fr\'echet distance between the two, since $\apath_0$ must conduct a self-crossing,
which increases the Fr\'echet distance by at least $2\delta$. Again, see \figref{dimembed}
\end{proof}

We conclude with additional details demonstrating the path-connectivity of balls for
embedded paths in $\R^4$ or higher.

\begin{lemma}[Balls in $(\epaths, d_{FP})$, $n \geq 4$]\label{lem:embed-path-balls}
    If $n \geq 4$, balls in the metric space of embedded paths
    $(\epaths, d_{FP})$ in $\R^n$ are path-connected.
\end{lemma}

\begin{proof}
    Let $\eqpath{0}, \eqpath{1}, \eqpath{2} \in \epaths$ in the ambient space $\R^n$, for $n \geq 4$.
    Let $\delta >0$, and $\B := \ball{d_{FG}}{\eqpath{0}}{\delta}\subset \epaths$.
    Since all topological knots are represented equivalently in only three dimensions,
    \journal{BTF: this sentence confuses me}
    without loss of generality, we consider the projections of every $\apath_0 \in \eqpath{0},
    \apath_1 \in \eqpath{1}$, and $\apath_2 \in \eqpath{2}$
    in $\R^3$. Construct a continuous
    $\Gamma: [0,1] \to \epaths$ by the linear interpolation from $\Gamma(0) =
    \apath_1$ to $\Gamma(1) = \apath_2$. By the rectifiability of the embeddings
    $\apath_1$ and $\apath_2$, the interpolation must reduce $d_{FP}(\apath_1, \apath_2)$ by some
    $\epsilon > 0$ before a self-crossing is required in the image of $\Gamma_t$
    at some $t \in [0,1]$.

    At $t$, conduct a self-crossing by perturbing $\Gamma_t$ in the fourth dimension by
    no more than $\epsilon/2$. This increases $d_{FP}(\Gamma_t, \apath_2)$ by no more than $\epsilon/2$.
    Hence, $d_{FP}(\Gamma_t, \apath_2)$ is either strictly decreasing as $t \to 1$, or necessarily satisfies
    $d_{FP}(\Gamma_t, \apath_2) \leq \delta - \epsilon/2$ for $\epsilon > 0$. This is to say, for all
    $t \in [0,1]$, $d_{FP}(\Gamma_t, \apath_2) \leq \delta$, and $\Gamma_t \in \B$. Hence, metric balls in the space are
    path-connected.
\end{proof}

\end{document}